\documentclass[11pt,letterpaper]{amsart}
\usepackage{amsmath,amsfonts,amssymb,amsthm,mathscinet,enumitem}
\usepackage{color}
\usepackage[table]{xcolor}
\usepackage{cite, hyperref}
\usepackage[noabbrev,capitalize]{cleveref}

\renewcommand{\phi}{\varphi}
\DeclareMathOperator{\len}{len}
\DeclareMathOperator{\supp}{supp}

\newcommand{\teopen}{[\kern-1.5pt [}
\newcommand{\teclose}{]\kern-1.5pt ]}
\newcommand{\eqopen}{[\kern-1.5pt [\kern-1.5pt [}
\newcommand{\eqclose}{]\kern-1.5pt ]\kern-1.5pt ]}
\newcommand{\ac}[1]{[{#1}]}
\newcommand{\fac}[1]{[{#1}]_F}
\newcommand{\te}[1]{\teopen{#1}\teclose}
\newcommand{\fte}[1]{\teopen{#1}\teclose_F}
\newcommand{\eq}[1]{\eqopen{#1}\eqclose}
\newcommand{\feq}[1]{\eqopen{#1}\eqclose_F}
\newcommand{\C}{{\mathbb C}}
\newcommand{\N}{{\mathbb N}}
\newcommand{\Q}{{\mathbb Q}}
\newcommand{\R}{{\mathbb R}}
\newcommand{\Z}{{\mathbb Z}}
\newcommand{\laurc}{\C[z,z^{-1}]}
\newcommand{\laurr}{\R[z,z^{-1}]}
\newcommand{\laurf}{F[z,z^{-1}]}
\newcommand{\laurfc}{\conj{F}[z,z^{-1}]}
\newcommand{\laurcu}{\C[z,z^{-1}]^\times}
\newcommand{\laurfu}{F[z,z^{-1}]^\times}
\newcommand{\laurfcu}{\conj{F}[z,z^{-1}]^\times}
\newcommand{\laure}{E[z,z^{-1}]}
\newcommand{\fa}{{\mathfrak a}}
\newcommand{\fb}{{\mathfrak b}}
\newcommand{\fc}{{\mathfrak c}}
\newcommand{\fd}{{\mathfrak d}}
\newcommand{\Cu}{\C^\times}
\newcommand{\Fu}{F^\times}
\newcommand{\Ru}{R^\times}
\newcommand{\conj}[1]{\overline{#1}}
\newcommand{\bigcups}[1]{\bigcup_{\substack{#1}}}
\newcommand{\ceil}[1]{\left\lceil{#1}\right\rceil}
\newtheorem{theorem}{Theorem}[section]
\newtheorem{proposition}[theorem]{Proposition}
\newtheorem{lemma}[theorem]{Lemma}
\newtheorem{corollary}[theorem]{Corollary}
\newtheorem{construction}[theorem]{Construction}
\newtheorem{openproblem}[theorem]{Open Problem}
\theoremstyle{definition}
\newtheorem{remark}[theorem]{Remark}

\begin{document}

\title{Sequences with identical autocorrelation functions}
\author{Daniel J. Katz}\thanks{This paper is based upon work supported in part by the National Science Foundation under Grants 1815487 and 2206454.
This research was done using services provided by the Open Science Grid Consortium \cite{osg07,osg09,osg06,osg15}, which is supported by the National Science Foundation awards 2030508 and 1836650.
Adeebur Rahman was supported by the Ramanujan Research Scholarship from the Department of Mathematics, California State University, Northridge.  Michael J Ward was supported by the Efrem Ostrow Scholarship from the Department of Mathematics, California State University, Northridge.}
\address{Department of Mathematics, California State University, Northridge, \: United States}
\author{Adeebur Rahman}
\address{Department of Mathematics, California State University, Northridge, \: United States}
\author{Michael J Ward}
\address{Department of Mathematics, California State University, Northridge, \: United States and University of California, Riverside}

\date{04 January 2025}
\begin{abstract}
Aperiodic autocorrelation is an important indicator of performance of sequences used in communications, remote sensing, and scientific instrumentation.
Knowing a sequence's autocorrelation function, which reports the autocorrelation at every possible translation, is equivalent to knowing the magnitude of the sequence's Fourier transform.
The phase problem is the difficulty in resolving this lack of phase information.
We say that two sequences are equicorrelational to mean that they have the same aperiodic autocorrelation function.
Sequences used in technological applications often have restrictions on their terms: they are not arbitrary complex numbers, but come from a more restricted alphabet.
For example, binary sequences involve terms equal to only $+1$ and $-1$.
We investigate the necessary and sufficient conditions for two sequences to be equicorrelational, where we take their alphabet into consideration.
There are trivial forms of equicorrelationality arising from modifications that predictably preserve the autocorrelation, for example, negating a binary sequence or reversing the order of its terms.
By a search of binary sequences up to length $44$, we find that nontrivial equicorrelationality among binary sequences does occur, but is rare.
An integer $n$ is said to be equivocal when there are binary sequences of length $n$ that are nontrivially equicorrelational; otherwise $n$ is unequivocal.
For $n \leq 44$, we found that the unequivocal lengths are $1$--$8$, $10$, $11$, $13$, $14$, $19$, $22$, $23$, $26$, $29$, $37$, and $38$.
We pose open questions about the finitude of unequivocal numbers and the probability of nontrivial equicorrelationality occurring among binary sequences.
\end{abstract}
\maketitle
\section{Introduction}
In many physical measurements of wave phenomena, detectors are unable to discern phases.  This loss of phase information is called the phase problem, a terminology that arose in x-ray crystallography, where a diffraction pattern gives the magnitude of the Fourier transform of the electron density without the phase information \cite{Bendory-Edidin}.
Knowing the magnitude of the Fourier transform is the same as knowing the autocorrelation of the electron density, which in general is not sufficient information to recover the electron density itself.
Classical x-ray crystallography provides the periodic autocorrelation of electron density of the contents of a unit cell of a crystal.
Modern imaging techniques have prompted researchers to also investigate phase retrieval in the aperiodic regime; see \cite[pp.~1491--1492]{Bendory-Edidin} and 
\cite{Shechtman-Eldar-Cohen-Chapman-Miao-Segev}.
This paper concerns itself with the aperiodic one-dimensional discrete problem of phases, that is, the extent to which one can deduce a sequence from its aperiodic autocorrelation.
Autocorrelation of sequences is important in many applications in communications and remote sensing where accurate timing and synchronization are required; see \cite{Golomb-1967, Golomb-Gong}.
As some examples, many foundational digital communications protocols such as code-division multiple access (CDMA) and orthogonal frequency-division multiplexing (OFDM) use low autocorrelation sequences, as do pulse compression schemes for efficient operation of radar.
When a CDMA system uses a sequence for modulation, the sequence's aperiodic autocorrelation determines its periodic and negaperiodic (also known as odd periodic) correlation functions, both which are important in determining the performance of the system \cite[Sec.\ V.B]{Sarwate-Pursley}.
Therefore, aperiodic autocorrelation can be viewed as the central object of interest in such systems.

Because we consider aperiodic autocorrelation, a {\it sequence} is any doubly infinite sequence $f=(\ldots,f_{-1},f_0,f_1,f_2,\ldots)$ of complex numbers such that only finitely many of the terms are nonzero.
We identify this sequence with $f(z)=\sum_{j \in \Z} f_j z^j \in \laurc$, where $\laurc$ is the ring of Laurent polynomials with complex coefficients.
Whenever we simply write a letter like ``$g$'' for a Laurent polynomial, it should be interpreted as shorthand for ``$g(z)$''.  Sometimes we write the full ``$g(z)$'' notation, especially when distinguishing $g(z)$ from other polynomials derived from $g(z)$ such as $g(-z)$ or $g(z^2)$.
If $R$ is any ring, then $R^\times$ denotes its group of units, and we say that two elements $f$ and $g$ are {\it $R$-associates} to mean that there is some $u \in R^\times$ such that $f= u g$.
Notice that the units of $\laurc$ are monomials with nonzero coefficients, that is, elements of the form $c z^j$ with $j \in \Z$ and $c \in \Cu$.
Multiplication by a unit in the Laurent polynomial formalism shifts and scales a sequence, which for our purposes produces an equivalent sequence, so we are only interested in sequences up to the relation of being $\laurc$-associates.

The {\it support} of a sequence $f$, written $\supp f$, is the set $\{j \in \Z: f_j\not=0\}$.
A {\it segment} is a set of consecutive integers.
The {\it length} of a sequence $f$, written $\len f$, is the cardinality of the smallest segment that contains $\supp f$.
A {\it contiguous sequence} $f$ is a sequence where $\supp f$ is a segment.
For any positive integer $m$, an {\it $m$-ary sequence} is a contiguous sequence where $f_j$ is an $m$th root of unity in $\C$ for every $j \in \supp f$; when $m=2$, we have a {\it binary sequence}, where $f_j \in \{1,-1\}$ for every $j \in \supp f$.
For readers more familiar with considering a binary sequence as a vector of $0$s and $1$s from the binary finite field $\Z/2\Z$, we remark that we can represent such a sequence as a vector of $+1$s and $-1$s in $\C$ by applying the transformation $x \mapsto (-1)^x$ to the terms.
Correlation measures the resemblance between two binary vectors by counting the number of coordinates where the vectors agree and deducting the number of coordinates where they disagree, so using the $+1/-1$ representation of binary sequences makes correlation equal to the dot product of the two vectors.
See \cite[p.~14]{Golomb-1994} for how the discussion of binary sequences passes naturally from $0/1$ representation to $+1/-1$ representation, and \cite[pp.~595--596]{Sarwate-Pursley} on why more general sequences whose terms lie in $\C$ are considered in communications systems.
For these more general sequences, correlation is still a dot product calculated in $\C$.
It is important that the correlation calculation happens in a ring of characteristic $0$ (such as $\C$) and not in a finite ring (such as $\Z/2\Z$, in which the terms of $0/1$-binary sequences lie) because finite rings have modular arithmetic, which would result in correlation values vanishing due to modular reduction in many cases where there is substantial agreement between sequences.
For example, $m$-ary sequences can be considered as vectors of elements of $\Z/m\Z$, but for the purposes of calculating correlation, one uses the map $x \mapsto \exp(2\pi i x/m)$ to transform them into sequences whose terms are complex $m$th roots of unity.

For a sequence $f=(\ldots,f_{-1},f_0,f_1,f_2,\ldots)$ and an integer $s$, the {\it aperiodic autocorrelation of $f$ at shift $s$} is
\[
C_f(s)=\sum_{j \in \Z} f_{j+s} \conj{f_j}.
\]
Note that the finite support of $f$ guarantees that $C_f(s)\not=0$ for only finitely many $s \in \Z$.
The Laurent polynomial interpretation of sequences provides a convenient formalism for calculating autocorrelation by thinking of sequences as functions on the complex unit circle.
To use this formalism, for $f(z)=\sum_{j \in \Z} f_j z^j \in \laurc$ we define the {\it conjugate of $f(z)$}, written $\conj{f(z)}$, to be $\sum_{j \in \Z} \conj{f_j} z^{-j}$.
(Notice that conjugation of a sequence reverses the order of its terms and then replaces each one with its complex conjugate.)
A {\it self-conjugate} element $f$ of $\laurc$ is one for which $\conj{f}=f$.
With conjugation defined,  one readily shows that
\[
f(z) \conj{f(z)} = \sum_{s \in \Z} C_f(s) z^s.
\]
We call $f\conj{f}$ the {\it autocorrelation function of $f$} because it organizes each autocorrelation value $C_f(s)$ as the coefficient of $z^s$, so that one can read off the autocorrelation value at any shift.
Notice that all autocorrelation functions are self-conjugate.

If $S \subseteq \laurc$, then the {\it conjugate of $S$}, written $\conj{S}$, is defined to be $\{\conj{s}: s \in S\}$.
We say that such a set $S$ is {\it self-conjugate} to mean $\conj{S}=S$.
If $F$ is a self-conjugate subfield of $\C$, then conjugation restricts to an automorphism of $F[z,z^{-1}]$ that is its own inverse.
One self-conjugate subfield of $\C$ is $\C$ itself, but other self-conjugate subfields include $\Q$ (where terms of binary sequences lie) and $\Q(e^{2\pi i/m})$ for each positive integer $m$, which is the field in which the terms of $m$-ary sequences lie.

We are interested in the extent to which the autocorrelation function $\sum_{s\in\Z} C_f(s) z^s$ determines the sequence $f$ from which it is derived.
We say that two sequences $f$ and $g$ are {\it equicorrelational} to mean that their autocorrelation functions are equal up to a positive real constant scalar multiple, i.e., $f\conj{f}=c g\conj{g}$ for some positive $c \in \R$.
This is equivalent to saying $f\conj{f}$ and $g\conj{g}$ are $\laurc$-associates (see \cref{Larry} for a proof), so equicorrelationality is an equivalence relation.
Note that since $\laurc$ is an integral domain, no nonzero sequence is equicorrelational to the zero sequence.
Since $C_f(0)$ is the squared Euclidean norm of the sequence $f$, two nonzero sequences are equicorrelational if and only if their normalizations with Euclidean norm $1$ have identical autocorrelation functions.

Associate sequences are equicorrelational (see \cref{Ursula} for a proof), and we should note that a sequence $f$ is equicorrelational to $\conj{f}$ (as well as any sequence associate to $\conj{f}$).
We say that two sequences in $f, g \in \laurc$ are {\it trivially equicorrelational} to mean that they are either associate to each other or one is associate to the conjugate of the other.
If $f$ and $g$ are equicorrelational but not trivially equicorrelational, we say that they are {\it nontrivially equicorrelational}.
This paper studies when nontrivial equicorrelationality can occur; when this happens, the autocorrelation function does not determine the sequence up to shifting, scaling, and conjugation (the last of which, it should be recalled, involves both reversal of the sequence and conjugation of every term).
Trivial equicorrelationality is an equivalence relation that refines equicorrelationality and is refined by the associate relation.

For $f \in \laurc$, the {\it associate class of $f$}, written $\ac{f}$, is the set of all $\laurc$-associates of $f$, so $\ac{f}=\{c z^j f: c \in \Cu, j \in \Z\}$.
The {\it trivial equicorrelationality class of $f$}, written $\te{f}$, is the set of all sequences that are trivially equicorrelational to $f$, so $\te{f}=[f]\cup [\conj{f}]$.
The {\it equicorrelationality class of $f$}, written $\eq{f}$, is the set of all sequences that are equicorrelational to $f$, and is a union of trivial equicorrelationality classes.
If $F$ is a self-conjugate subfield of $\C$ and $f \in \laurf$, then the {\it $\laurf$-associate class of $f$}, written $\fac{f}$, is the set of all associates of $f$ in $\laurf$, which is just $\ac{f} \cap \laurf$ (see \cref{Ernie} for a proof).
The {\it $F$-trivial equicorrelationality class of $f$}, written $\fte{f}$, is the set of all sequences in $\laurf$ that are trivially equicorrelational to $f$, that is, $\te{f} \cap \laurf$, which is equal to $\fac{f} \cup \fac{\conj{f}}$ (see \cref{Felix} for a proof).
The {\it $F$-equicorrelationality class of $f$}, written $\feq{f}$, is the set of all sequences in $\laurf$ that are equicorrelational to $f$, that is, $\eq{f} \cap \laurf$.

A {\it generalized palindrome} is a sequence $f$ in $\laurc$ that is a $\laurc$-associate of its own conjugate, that is, $\conj{f} \in \ac{f}$.
If $F$ is a self-conjugate subfield of $\C$, each $\laurf$-associate class either consists entirely of generalized palindromes and is self-conjugate, or else the class has no generalized palindromes and its conjugate is a different $\laurf$-associate class (see \cref{Herbert} for a proof).

We are interested in how the alphabet of values that can occur as sequence terms influences equicorrelationality.
For a given sequence $f$, many of the sequences that are equicorrelational to $f$ might have terms that do not reside in the same alphabet that was used to construct $f$.
Our first result shows how restriction of sequence terms to a self-conjugate subfield of $\C$ constrains the possibilities for equicorrelationality.
In the following theorem, we use the fact that Laurent polynomial rings over fields are unique factorization domains, and throughout this paper, we set $\N=\{0,1,2\ldots\}$.
\begin{theorem}\label{Barbara}
Let $F$ be a self-conjugate subfield of $\C$ and $f \in \laurf$.
If $f=0$, then $\feq{f}=\fte{0}=\fac{0}=\{0\}$.
If $f\not=0$, then suppose that
\begin{equation}\label{Daisy}
f=u f_1^{a_1}\cdots f_m^{a_m} g_1^{b_1} \cdots g_n^{b_n} \conj{g_1}^{c_1} \cdots \conj{g_n}^{c_n}
\end{equation}
is a factorization of $f$ into nonassociate $\laurf$-irreducibles $f_1$, $\ldots$, $f_n$, $g_1$, $\ldots$, $g_n$, $\conj{g_1}$, $\ldots$, $\conj{g_n}$ and unit $u$ of $\laurf$ where $f_1$, $\ldots$, $f_m$ are generalized palindromes and $g_1$, $\ldots$, $g_n$ are not, and we have $a=(a_1,\ldots,a_m) \in \N^m$ and $b=(b_1,\ldots,b_n), c=(c_1,\ldots,c_n) \in \N^n$.
Then
\begin{align}
\feq{f} 
& = \bigcups{b',c' \in \N^n \\ b'+c'=b+c} \left[ f_1^{a_1}\cdots f_m^{a_m} g_1^{b_1'} \cdots g_n^{b_n'} \conj{g_1}^{c_1'} \cdots \conj{g_n}^{c_n'} \right]_F \label{Rosa} \\
& = \bigcups{b',c' \in \N^n \\ b'+c'=b+c \\ b' \leq c'} \left[\kern-3pt \left[ f_1^{a_1}\cdots f_m^{a_m} g_1^{b_1'} \cdots g_n^{b_n'} \conj{g_1}^{c_1'} \cdots \conj{g_n}^{c_n'}\right] \kern-3pt \right]_F, \label{Rose}
\end{align}
where the $b' \leq c'$ is using the lexicographic ordering of $\N^n$.
Let $N=\prod_{j=1}^n (b_j+c_j+1)$.
The  union in \eqref{Rosa} is of $N$ pairwise disjoint $\laurf$-associate classes and the union in \eqref{Rose} is of $\ceil{N/2}$ pairwise disjoint $F$-trivial equicorrelationality classes.
The count $N$ is odd if and only if $b_j+c_j$ is even for every $j \in \{1,2,\ldots,n\}$.
When $N$ is odd, precisely one of the $\laurf$-associate classes in \eqref{Rosa} is self-conjugate and precisely one of the $F$-trivial equicorrelationality classes in \eqref{Rose} is composed of a single $\laurf$-associate class (namely, the self-conjugate $\laurf$-associate class just mentioned); no such classes occur in \eqref{Rosa} or \eqref{Rose} when $N$ is even.
All the other $\laurf$-associate classes in \eqref{Rosa} are non-self-conjugate and occur in conjugate pairs, and all other $F$-trivial equicorrelationality classes in \eqref{Rose} contain two $\laurf$-associate classes (which are conjugate pairs) each.
The sequence $f$ is nontrivially equicorrelational to some other sequence in $\laurf$ if and only if $N \geq 3$.
\end{theorem}
\begin{remark}
When $F=\C$ in \cref{Barbara}, we can factor $f$ completely into linear factors (times a unit), and then one recapitulates the results described in Theorem 2.4 of \cite{Beinert-Plonka}, which obtains results already shown in \cite{Fejer}.  If $f$ represents a sequence of length $\ell \geq 2$, then we can obtain $\ell-1$ linear factors in \eqref{Daisy} and so the number of nontrivial equicorrelationality classes in the equicorrelationality class of $f$ is at most $2^{\ell-2}$, as observed in \cite[Cor.\ 2.6]{Beinert-Plonka}.  We note that the maximum of $2^{\ell-2}$ is achieved if and only if either (i) $m=0$, $n=\ell-1$, and $\{b_j,c_j\}=\{0,1\}$ for every $j \in \{1,\ldots,\ell-1\}$ or (ii) $\ell=3$ with $m=0$, $n=1$, and $b_1+c_1=2$.  
\end{remark}
\begin{remark}
Beinert and Plonka \cite[Remark 2.7]{Beinert-Plonka} also consider what happens when $F=\R$, the real field, in the situation outlined in \cref{Barbara}, and (if we translate their result into the language of this paper) they point out that a real sequence written as a polynomial $f$ of length $\ell$ can have $2^{\ell-2}$ nontrivial equicorrelationality classes in its equicorrelationality class only if all its roots are real (i.e., if and only if $f$ splits in $\R[z]$).
\end{remark}
\cref{Barbara} limits the circumstances under which generalized palindromes may be equicorrelational to each other, as we shall show when we prove the following corollary.
\begin{corollary}\label{David}
If $f$ and $g$ are generalized palindromes that are equicorrelational, then they must be $\laurc$-associates.
\end{corollary}
Furthermore, we show that certain kinds of generalized palindromes cannot be equicorrelational to each other.
A {\it palindrome} is a sequence $f \in \laurr$ such that $\conj{f}=z^j f$ for some $j \in \Z$.
An {\it antipalindrome} is a sequence $f \in \laurr$ such that $\conj{f}=-z^j f$ for some $j \in \Z$.
Palindromes and antipalindromes are the only kinds of generalized palindromes that occur among the binary sequences, and the only sequence that is both a palindrome and an antipalindrome is the zero sequence.
We shall prove the following as a consequence of \cref{Barbara}.
\begin{corollary}\label{Francis}
It is not possible for a palindrome in $\R[z,z^{-1}]$ to be equicorrelational to an antipalindrome in $\R[z,z^{-1}]$ unless both the sequences are $0$.
\end{corollary}
For the rest of this introduction, we restrict the relation of equicorrelationality to binary sequences: an equivalence class of this relation is called a {\it binary equicorrelationality class}.
We also restrict the notion of trivial equicorrelationality to binary sequences: two binary sequences $f$ and $g$ are trivially equicorrelational if and only if $f=u z^j g$ or $f=u z^j \conj{g}$ for some $u \in \{-1,1\}$ and $j \in \Z$, and an equivalence class of this relation is called a {\it trivial binary equicorrelationality class}.
Trivial binary equicorrelationality refines binary equicorrelationality, so every binary equicorrelationality class is a union of pairwise disjoint trivial binary equicorrelationality classes.
The {\it volume} of a binary equicorrelationality class equals the number of trivial binary equicorrelationality classes in this union, and a binary equicorrelationality class with volume greater than one is called {\it nontrivial}.

We used a computer program to find all nontrivial binary equicorrelationality classes for binary sequences of lengths $1$ through $44$.
The searches for lengths $35$ and larger were made using opportunistic grid computing resources provided by the Open Science Grid Consortium \cite{osg07,osg09,osg06,osg15}.
The program was written primarily in the Rust language, with some use of the C language to enable the program to use the polynomial factorization routine in the PARI library \cite{PARI}, which in turn depends on the GNU Multiple Precision Arithmetic Library \cite{GMP}.
In \cref{George}, we indicate how many nontrivial binary equicorrelationality classes there are of each volume.
We represent the distribution of volumes of nontrivial equicorrelationality classes in a compact notation $n_1 [v_1] + n_2 [v_2]+\cdots+n_t [v_t]$, which means that there are $n_i$ classes of volume $v_i$ for each $i\in \{1,2,\ldots,t\}$.
If an entry for a particular sequence length is blank, it means that there are no nontrivial equicorrelationality classes for binary sequences of that length.
One can see that we did not encounter any nontrivial binary equicorrelationality class of odd volume.
It is also noteworthy that we did not discover any nontrivial binary equicorrelationality class that contains a palindrome or antipalindrome.

\begin{table}[ht!]
\caption{Nontrivial binary equicorrelationality classes}\label{George}
\begin{center}
\rowcolors{2}{gray!15}{white}
\begin{tabular}{c|c||c|c}
sequence & frequency [volume] & sequence & frequency [volume] \\
length & of nontrivial classes & length & of nontrivial classes \\
\hline
1 &         & 23 &                   \\ 
2 &         & 24 & 422 [2]           \\ 
3 &         & 25 & 36 [2]            \\ 
4 &         & 26 &                   \\ 
5 &         & 27 & 348 [2] + 1 [4]   \\ 
6 &         & 28 & 180 [2]           \\                   
7 &         & 29 &                   \\                
8 &         & 30 & 1214 [2]          \\                 
9 & 1 [2]   & 31 & 26 [2]            \\                  
10 &        & 32 & 1136 [2]          \\         
11 &        & 33 & 1105 [2]          \\                
12 & 8 [2]  & 34 & 30 [2]            \\                        
13 &        & 35 & 349 [2]           \\              
14 &        & 36 & 8230 [2] + 16 [4] \\
15 & 14 [2] & 37 &                   \\          
16 & 12 [2] & 38 &                   \\
17 & 1 [2]  & 39 & 4102 [2]          \\
18 & 42 [2] & 40 & 6288 [2]          \\
19 &        & 41 & 4[2]              \\                 
20 & 44 [2] & 42 & 17574 [2]         \\
21 & 67 [2] & 43 & 22 [2]            \\
22 &        & 44 & 3104 [2]          \\    
\end{tabular}
\end{center}
\end{table}

We define a binary sequence $f$ to be {\it equivocal} if it is nontrivially equicorrelational to some other binary sequence; otherwise $f$ is {\it unequivocal}.
A positive integer $n$ is said to be {\it equivocal} if there is an equivocal binary sequence of length $n$; otherwise $n$ is {\it unequivocal}.
\cref{George} shows that the numbers from $1$ to $8$, along with $10$, $11$, $13$, $14$, $19$, $22$, $23$, $26$, $29$, $37$, and $38$ are unequivocal.
We shall prove the following result, which explains why many numbers are equivocal.
\begin{proposition}\label{Karen}
Let $m$, $n$ be positive integers such that $m | n$. If $m$ is equivocal, then $n$ is equivocal.
\end{proposition}
Perusal of \cref{George} shows that unequivocal numbers seem to become more sparse as the length increases.
This leads to the following open question.
\begin{openproblem}
Are there finitely or infinitely many unequivocal numbers?
\end{openproblem}
Further perusal of \cref{George} shows that the number of nontrivial equicorrelationality classes sometimes increases as sequence length increases, but when one considers that the total number of binary sequences doubles every time the length increases by $1$, the fraction of equivocal sequences does not appear to be on a trend of growth.
This suggests another open question.
\begin{openproblem}
Does the fraction of equivocal binary sequences vanish asymptotically?  That is, if we define $N_\ell$ to be the number of equivocal binary sequences of length $\ell$, does $N_\ell/2^\ell$ tend to $0$ as $\ell$ tends to infinity?  
\end{openproblem}
The rest of this paper is organized as follows.
\cref{Helen} contains preliminaries of notations and basic results.
In \cref{Jason}, we prove \cref{Barbara}.
In \cref{Alexandra}, we prove its Corollaries \ref{David} and \ref{Francis}.
In \cref{Hilda}, we prove \cref{Karen}.

\section{Preliminaries}\label{Helen}

Throughout this paper, $\N=\{0,1,2,\ldots\}$ and if $R$ is a ring, then $\Ru$ denotes the unit group of $R$.
We always use the Laurent polynomial formalism for sequences and their autocorrelation functions, as described in the Introduction, so sequences are always thought of as elements of the Laurent polynomial ring $\laurc$.
We retain the convention that whenever we simply write a letter like ``$g$'' for a Laurent polynomial, it should be interpreted as shorthand for ``$g(z)$'', but we sometimes write the full ``$g(z)$'' notation, especially when distinguishing $g(z)$ from other polynomials derived from $g(z)$ such as $g(-z)$ or $g(z^2)$.
For any field $F$, we have $\laurfu=\{c z^j: c \in \Fu, j \in \Z\}$, so every nonzero $f \in \laurf$ can be written uniquely as $f=u g$ where $u \in \laurfu$ and $g$ is a monic polynomial in $F[z]$ with a nonzero constant coefficient, and then the length (originally described in the third paragraph of the Introduction) of $f$ is $1+\deg g$.
(And, of course, $\len(0)=0$.)
In particular, an element of $\laurf$ is a unit if and only if it has length $1$.
Then one can verify that $\laurf$ is a Euclidean domain with $\len$ as its Euclidean size function.
In particular, $\len(f g) \geq \len f$ for all $f, g \in \laurf$ with $g\not=0$, and thus $\laurf$-associates have the same length.
Since $\laurf$ is a Euclidean domain, every pair of elements $f,g \in \laurf$ has a greatest common divisor (defined up to $\laurf$-associates), and if $f$ and $g$ are not both $0$, then we shall use $\gcd(f,g)$ to denote the unique monic polynomial with nonzero constant coefficient in the $\laurf$-associate class containing the greatest common divisors (and we set $\gcd(0,0)=0$).

A few other more precise results about units, the length function, and greatest common divisors will be useful later in the paper.

\begin{lemma}\label{Leonard}
If $F$ is a field and if $f,g$ are nonzero elements of $\laurf$, then $\len(f g)=\len f+\len g-1$.
\end{lemma}
\begin{proof}
Write $f=u a$ and $g=v b$ with $u,v \in \laurfu$ and $a,b$ monic polynomials with nonzero constant coefficients.  Then $f g=(u v)(a b)$ and $u v$ is a unit while $a b$ is a monic polynomial with a nonzero constant coefficient.  Thus, $\len(f g)=1+\deg(a b)=(1+\deg a)+(1+\deg b)-1=\len f+\len g-1$.
\end{proof}

\begin{lemma}\label{Dorothy}
Let $F$ be a field, $u(z) \in \laurfu$, and $m \in \Z$.  Then $u(z^m) \in \laurfu$.
\end{lemma}
\begin{proof}
Since $u(z)$ is a unit in $\laurf$, there is some $v(z) \in \laurf$ such that $u(z)v(z)=1$.  Substituting $z^m$ for $z$ in this expression shows that $u(z^m)$ is a unit in $\laurf$.
\end{proof}

\begin{lemma}\label{Samuel}
Let $F$ be a field, let $f(z)$ be a nonzero element of $\laurf$, and let $m$ be a nonzero integer.
Then $\len f(z^m)=(-1+\len f(z))|m|+1$.
\end{lemma}
\begin{proof}
Write $f(z)=u(z) a(z)$ where $u(z) \in \laurfu$ and $a(z)$ is a monic polynomial in $F[z]$ with a nonzero constant coefficient with $\deg a(z)=-1+\len f(z)$.
Then $f(z^m)=u(z^m) a(z^m)$ and $u(z^m)$ is a unit in $\laurf$ by \cref{Dorothy}.
If $m$ is positive, then $a(z^m)$ is a monic polynomial in $F[z]$ with nonzero constant coefficient and degree $m \deg a(z)$, so $\len f(z^m)=1+\deg a(z^m)=1+m\deg a(z) =1+|m|(-1+\len f(z))$.
If $m$ is negative, let $a_0$ be the constant coefficient of $a(z)$, and then $a(z^m)=a_0 z^{m \deg a(z)} b(z)$ where $b(z)$ is a monic polynomial in $F[z]$ with nonzero constant coefficient and degree $-m \deg a(z)$; then $f(z^m)$ is an associate of $b(z)$, so $\len f(z)=1+\deg b(z)=1+|m|(-1+\len f(z))$.
\end{proof}

\begin{lemma}\label{Conrad}
Let $F$ be a field, let $f(z)$ and $g(z)$ be relatively prime elements of $\laurf$, and let $m \in \Z$.
Then $f(z^m)$ and $g(z^m)$ are also relatively prime elements of $\laurf$.
\end{lemma}
\begin{proof}
Since $\laurf$ is a Euclidean domain, it is a principal ideal domain, so there are $a(z),b(z) \in \laurf$ such that $a(z) f(z) + b(z) g(z)=1$.  Then $a(z^m) f(z^m) + b(z^m) g(z^m)=1$ also, showing that $f(z^m)$ and $g(z^m)$ are relatively prime.
\end{proof}

In the Introduction, it was claimed that if $F$ is a subfield of $\C$ and $f \in \laurf$, then $\fac{f}=\ac{f}\cap\laurf$, that is, the set of $\laurf$-associates of $f$ is obtained from the set $\ac{f}$ of $\laurc$-associates of $f$ by just taking those elements in the latter set whose coefficients all lie in $F$.  We prove a slightly more general result here (where $\C$ is replaced with an arbitrary extension field $E$ of $F$).
\begin{lemma}\label{Ernie}
Let $E$ be a field and let $F$ be a subfield of $E$.  Let $f,g \in \laurf$.  Then $f$ and $g$ are $\laure$-associates if and only if they are $\laurf$-associates.  Thus, if $F$ is a subfield of $\C$, then $[f]_F=[f]\cap\laurf$.
\end{lemma}
\begin{proof}
Suppose that $f$ and $g$ are $\laurf$-associates.  Since every unit of $\laurf$ is a unit of $\laure$, we see that $f$ and $g$ are $\laure$-associates.

Now suppose that $f$ and $g$ are $\laure$-associates.  If one of $f$ or $g$ is zero, then the other must also be zero, and then they are clearly $\laurf$-associates.  So we may assume that $f$ and $g$ are nonzero from now on, and then there is some unit $u$ of $\laure$ such that $f=u g$.  Now $u=e z^j$ for some nonzero $e \in E$.  If $f_k z^k$ is the lowest degree monomial in $f$ (so $f_k \in F^\times$), then $e f_k z^{k+j}$ is the lowest degree monomial in $g$ (so $e f_k \in F^\times$).  Thus $e=(e f_k)/f_k \in F^\times$.  This makes $u$ a unit in $\laurf$, and so $f$ and $g$ are $\laurf$-associates.
This proves the first claim in this lemma.

The first claim of this lemma (applied with $E=\C$) shows that $h \in \laurc$ is an $\laurf$-associate of $f$ if and only if $h$ is both a $\laurc$-associate of $f$ and $h \in \laurf$, which means that $[f]_F=[f]\cap\laurf$.
\end{proof}

For any subfield $F$ of $\C$, the conjugation map $f \mapsto \conj{f}$ from $\laurf$ to $\laurfc$ is an isomorphism of rings and so maps $0$, units, irreducibles, and reducible elements of $\laurf$ respectively to $0$, units, irreducibles, and reducible elements of $\laurfc$.  In particular $\conj{\laurfu}=\laurfcu$.
This conjugation map also carries pairs of elements that are $\laurf$-associates to pairs of elements that are $\laurfc$-associates, and so $\conj{\fac{f}}=\ac{\conj{f}}_{\conj{F}}$.
First we show that conjugation preserves length.
\begin{lemma}\label{Lester}
If $f\in\laurc$, then $\len \conj{f}=\len f$.
\end{lemma}
\begin{proof}
If $f=0$, then $\conj{f}=0$, so $\len \conj{f}=\len f$, so from now on assume $f\not=0$.
Write $f= u g$ where $u \in \laurcu$ and $g$ is a monic polynomial with a nonvanishing constant coefficient $g_0$, so that $\len f=1+\deg g$.
Then $\conj{f} = \conj{u} \conj{g} = \conj{u} \conj{g_0} z^{-\deg g} (\conj{g_0}^{-1} z^{\deg g} \conj{g})$, but $\conj{u} \conj{g_0} z^{-\deg g} \in \laurcu$ and $\conj{g_0}^{-1} z^{\deg g} \conj{g}$ is a monic polynomial of degree $\deg g$ with nonvanishing constant coefficient $\conj{g_0}^{-1}$.  So $\len \conj{f}=1+\deg g=\len f$.
\end{proof}
The next technical result shows how associate generalized palindromes are related to each other.
\begin{lemma}\label{Laura}
Let $f$ and $g$ be generalized palindromes with $\conj{f}=u z^j f$ and $\conj{g}=v z^k g$ for some $u, v \in \Cu$ and $j,k \in \Z$.
Let $F$ be a subfield of $\C$ and suppose that $f,g \in \laurf$ and that $f$ and $g$ are $\laurc$-associates.
Then either (i) $f=g=0$ or else (ii) $v/u = \conj{w}/w$ for some $w \in \Fu$ and $j \equiv k \pmod{2}$.
\end{lemma}
\begin{proof}
Since $f$ and $g$ are $\laurc$-associates, either $f=g=0$ or both $f$ and $g$ are nonzero; we are done in the former case, so assume that the latter case holds.
Then there is some $w \in \Cu$ and $i \in \Z$ such that $g=w z^i f$.
We conjugate both sides to obtain $\conj{g} = \conj{w} z^{-i} \conj{f}$ and substitute the expressions for $\conj{g}$ and $\conj{f}$ from the statement of the lemma to obtain $v z^k g = \conj{w} z^{-i} u z^j f$, and then replace $g$ with $w z^i f$ again to obtain $v z^k w z^i f = \conj{w} z^{-i} u z^j f$.  Since $F[z,z^{-1}]$ is an integral domain, we can cancel out the nonzero term $f$ to obtain $v w z^{i+k} = u \conj{w} z^{j-i}$, and matching constants and exponents produces $v/u=\conj{w}/w$ and $j=k+2 i$.
\end{proof}

Now we present some basic results on equicorrelationality.
\begin{lemma}
If $f,g \in \laurc$ are equicorrelational, then $\len f = \len g$.
\end{lemma}
\begin{proof}
If $f$ and $g$ are equicorrelational, then $f\conj{f}=c g\conj{g}$ for some positive real number $c$, so by \cref{Leonard} we have $\len f + \len \conj{f}-1=\len c + \len g+\len \conj{g} -2$, and $\len c=1$ since $c$ is a nonzero constant.  Then by \cref{Lester}, we have $2 \len f -1 = 2 \len g-1$, so $\len f=\len g$.
\end{proof}

\begin{lemma}\label{Ursula}
Let $f$ and $g$ be $\laurc$-associates.  Then $f$ is equicorrelational to $g$.
\end{lemma}
\begin{proof}
We have $f=u g$ for some unit $u$ in $\laurc$, and $u=c z^j$ for some $c \in \Cu$.  Then $f\conj{f}=u\conj{u} g \conj{g}$ and $u\conj{u}=|c|^2$, which is a positive real number, and hence $f$ and $g$ are equicorrelational.
\end{proof}

\begin{lemma}\label{Larry}
Let $f, g \in \laurc$.  Then $f$ is equicorrelational to $g$ if and only if $f\conj{f}$ and $g\conj{g}$ are $\laurc$-associates.
\end{lemma}
\begin{proof}
Suppose that $f$ is equicorrelational to $g$, so that $f\conj{f}=c g\conj{g}$ for some positive real constant $c$.  This $c$ is a unit in $\laurc$, so $f\conj{f}$ and $g\conj{g}$ are $\laurc$-associates.

Now suppose that $f\conj{f}$ and $g\conj{g}$ are $\laurc$-associates.  If either of $f$ or $g$ is $0$, then both must be $0$, and then $f \conj{f} = 1 g \conj{g}$, so that $f$ and $g$ are clearly equicorrelational.  So from now on, we may assume that $f$ and $g$ are nonzero.  So there is some unit $u$ in $\laurc$ such that $f\conj{f}=u g\conj{g}$.
Conjugating both sides yields $f\conj{f}=\conj{u} g\conj{g}$, and since $g$ is nonzero (so $\conj{g}$ is nonzero) and $\laurc$ is an integral domain, we see that $u$ is self-conjugate.
Recall that the units of $\laurc$ are elements of the form $c z^j$ where $c \in \Cu$ and $j \in \Z$, so the self-conjugate units of $\laurc$ are just the nonzero real numbers.
So $u$ is a nonzero real number.
The constant coefficient of $f\conj{f}$ (resp., $g\conj{g}$) is the squared Euclidean norm of the sequence $f$  (resp., $g$) and the former is obtained from the latter by multiplying by the nonzero real number $u$.
Since $f$ and $g$ are nonzero sequences, these constant coefficients of their autocorrelation functions are positive real numbers, and this forces $u$ to be a positive real number, thus making $f$ and $g$ equicorrelational.
\end{proof}
The last results of this section examine the structure of $\laurf$-associate classes and $F$-trivial equicorrelationality classes.
\begin{lemma}\label{Herbert}
If $F$ is a self-conjugate subfield of $\C$ and $f \in \laurf$, then $\conj{\fac{f}}=\fac{\conj{f}}$.
If $f$ is a generalized palindrome, then $\fac{f}$ is self-conjugate and contains only generalized palindromes.
But if $f$ is not a generalized palindrome, then $\fac{f}$ is not self-conjugate and contains no generalized palindromes.
\end{lemma}
\begin{proof}
Recall that $\conj{\fac{f}}=[\conj{f}]_{\conj{F}}$, but since $F$ is self-conjugate, this means that $\conj{\fac{f}}=\fac{\conj{f}}$.
Note that $f$ is a generalized palindrome if and only if $\conj{f} \in \ac{f}$, which is true if and only if $\ac{\conj{f}}=\ac{f}$, which (because both $f$ and $\conj{f}$ lie in $\laurf$ since $F$ is self-conjugate) is true by \cref{Ernie} if and only if $\fac{\conj{f}}=\fac{f}$, which (by what we just proved) is equivalent to $\conj{\fac{f}}=\fac{f}$, which is the same as saying that $\fac{f}$ is self-conjugate.
The fact that an $\laurf$-associate class is self-conjugate if and only if an arbitrary representative is a generalized palindrome means that an $\laurf$-associate class cannot have some element that is a generalized palindrome and another element that is not.
\end{proof}
Now we present a result on how $F$-trivial equicorrelationality classes are related to $\laurf$-associate classes, and how these classes behave when they contain generalized palindromes.
\begin{lemma}\label{Felix}
Let $F$ be a self-conjugate subfield of $\C$ and $f \in \laurf$.
Then $\fte{f}$ is self-conjugate and $\fte{f}=\fac{f}\cup\fac{\conj{f}}$.
If $f$ is a generalized palindrome, then every element of $\fte{f}$ is a generalized palindrome, and $\fte{f}=\fac{f}=\fac{\conj{f}}$.
If $f$ is not a generalized palindrome, then no element of $\fte{f}$ is a generalized palindrome, and $\fte{f}$ is the union of two disjoint non-self-conjugate $\laurf$-associate classes, $\fac{f}$ and $\fac{\conj{f}}$, which are conjugates of each other.
\end{lemma}
\begin{proof}
By definition, $\fte{f}=\te{f}\cap\laurf$.  But $\te{f}=\ac{f}\cup\ac{\conj{f}}$, so
\begin{align}
\begin{split}\label{Jennifer}
\fte{f}
& =\left(\ac{f}\cup\ac{\conj{f}}\right)\cap\laurf\\
& =\left(\ac{f}\cap\laurf\right)\cup\left(\ac{\conj{f}}\cap\laurf\right) \\
& = [f]_F \cup [\conj{f}]_F,
\end{split}
\end{align}
where the third equality is from \cref{Ernie}.
Then we conjugate \eqref{Jennifer} to obtain $\conj{\fte{f}}=\conj{[f]_F \cup [\conj{f}]_F}=\conj{\fac{f}}\cup\conj{\fac{\conj{f}}}$.
By \cref{Herbert} we obtain $\conj{\fte{f}}=[\conj{f}]_F \cup [f]_F$, and then another application of \eqref{Jennifer} shows that $\conj{\fte{f}}=\fte{f}$, so $\fte{f}$ is self-conjugate.

If $f$ is a generalized palindrome, then $\conj{f} \in \ac{f}$, so then $\ac{\conj{f}}=\ac{f}$, and so $\fac{\conj{f}}=\fac{f}$, so then \eqref{Jennifer} shows that $\fac{\conj{f}}=\fac{f}=\fte{f}$.
Furthermore, \cref{Herbert} shows that every element of $\fac{f}$ is a generalized palindrome.

On the other hand, if $f$ is not a generalized palindrome, then $\conj{f}\not\in\ac{f}$, so $\ac{\conj{f}}$ must be disjoint from $\ac{f}$ since these are classes of an equivalence relation.
Thus, $\fac{\conj{f}}$ must be disjoint from $\fac{f}$.
Since $f$ is not a generalized palindrome, we know that $\conj{f}$ is not a generalized palindrome, and then \cref{Herbert} says that neither $\fac{f}$ nor $\fac{\conj{f}}$ is self-conjugate, nor does either of these two contain a generalized palindrome.
Thus, \eqref{Jennifer} shows that $\fte{f}$ contains no generalized palindrome.
We also note that the two $\laurf$-associate classes, $\fac{f}$ and $\fac{\conj{f}}$, are conjugates of each other by \cref{Herbert}.
\end{proof}

\section{Proof of \cref{Barbara}}\label{Jason}
In this section, we prove \cref{Barbara} from the Introduction.  The statements about what happens when $f=0$ arise from the observation made earlier in the Introduction that no nonzero sequence can be equicorrelational to the zero sequence, so we assume that $f\not=0$ henceforth.
Let $h$ be a sequence in $F[z,z^{-1}]$.
Then $h$ is equicorrelational to $f$ if and only if $h\conj{h} = t f\conj{f}$ for some positive real number $t \in F$.
Therefore, if $h$ is equicorrelational to $f$, then the unique factorization of $h$ can only contain the irreducibles in $f\conj{f}$.
Because $f_1,\ldots,f_m$ are generalized palindromes (hence associate to their own conjugates), in searching for the sequences equicorrelational to $f$ we can confine ourselves to sequences $h$ that can be written as 
\[
h = v f_1^{a_1'}\cdots f_m^{a_m'} g_1^{b_1'} \cdots g_n^{b_n'} \conj{g_1}^{c_1'} \cdots \conj{g_n}^{c_n'},
\]
for some unit $v \in F[z,z^{-1}]$ and $a'=(a_1',\ldots,a_m') \in \N^m$ and $b'=(b_1',\ldots,b_n')$, $c'=(c_1',\ldots,c_n') \in \N^n$.
For such a sequence $h$, the product $h\conj{h}$ has a unique factorization with $2 a_i'$ factors of each $f_i$ as well as $b_j' + c_j'$ factors of each $g_j$ and $b_j'+c_j'$ factors of each $\conj{g_j}$.
Meanwhile, $f\conj{f}$ has $2 a_i$ factors of each $f_i$ as well as $b_j + c_j$ factors of each $g_j$ and $b_j+c_j$ factors of each $\conj{g_j}$.
So $h$ is equicorrelational to $f$ if and only if $a'=a$ and $b'+c'=b+c$, which is true if and only if $h$ is in the union on the right-hand side of \eqref{Rosa}.
This union is of pairwise disjoint classes because the representatives that we have written for the $\laurf$-associate classes in the union are all non-$\laurf$-associates of each other.

The number of $\laurf$-associate classes in the union from \eqref{Rosa} equals the number of pairs $(b',c') \in \N^n \times \N^n$ such that $b'+c'=b+c$.
This last constraint forces $b' \in \prod_{j=1}^n \{0,1,\ldots,b_j+c_j\}$, and for each such $b'$, there is a unique $c'=b+c-b' \in \N^n$ such that $b'+c'=b+c$, so we have precisely 
\[
N=\prod_{j=1}^n |\{0,1,\ldots,b_j+c_j\}|=\prod_{j=1}^n (b_j+c_j+1)
\]
classes. 

The conjugate of a representative
\[
r=f_1^{a_1}\cdots f_m^{a_m} g_1^{b_1'} \cdots g_n^{b_n'} \conj{g_1}^{c_1'} \cdots \conj{g_n}^{c_n'}
\]
of one of the $\laurf$-associate classes in \eqref{Rosa} is $\conj{r}=w s$, where
\[
s= f_1^{a_1}\cdots f_m^{a_m} g_1^{c_1'} \cdots g_n^{c_n'} \conj{g_1}^{b_1'} \cdots \conj{g_n}^{b_n'}
\]
and where $w$ is some unit in $\laurf$ because each of $f_1,\ldots,f_m$ is a generalized palindrome.
Thus, by \cref{Herbert}, the conjugate of $[r]_F$ is $[\conj{r}]_F=[s]_F$, which means that the $\laurf$-associate class in \eqref{Rosa} indexed by $(b',c')$ is the conjugate of the class indexed by $(c',b')$.
Thus, $[r]_F$ is self-conjugate if and only if $b'=c'$.
This can be true of only one class (the one with $b'=c'=(b+c)/2$ if $b_j+c_j$ is even for every $j$) or none at all (if $b_j+c_j$ is odd for at least one $j$).
Hence, if $N$ is odd, then precisely one $\laurf$-associate class in \eqref{Rosa} is self-conjugate, but if $N$ is even, then no such class is self-conjugate, and in either case, the rest of the $\laurf$-associate classes occur in conjugate pairs.

Recall from \cref{Felix} that the $F$-trivial equicorrelationality class of $r$ is $\fac{r} \cup \fac{\conj{r}}$, which is either the union of two disjoint non-self-conjugate $\laurf$-associate classes or else is equal to a single self-conjugate $\laurf$-associate class.
Thus, when we pair up the class in \eqref{Rosa} indexed by $(b',c')$ with the one indexed by $(c',b')$, we produce a single $F$-trivial equicorrelationality class in \eqref{Rose}, and so in order to make \eqref{Rose} a union of pairwise disjoint classes, we impose the condition $b' \leq c'$.
Since every $F$-trivial equicorrelationality class in \eqref{Rose} arises from two $\laurf$-associate classes in \eqref{Rosa} (with the exception of the single self-conjugate $\laurf$-associate class that occurs when $N$ is odd---this single class is itself also an $F$-trivial equicorrelationality class), the number of $F$-trivial equicorrelationality classes in \eqref{Rose} is $\ceil{N/2}$.
Then $f$ is nontrivially equicorrelational to some other sequence in $\laurf$ if and only if \eqref{Rose} is a union of more than one $F$-trivial equicorrelationality class.
This happens if and only if $\ceil{N/2} > 1$, i.e., if and only if $N \geq 3$.\hfill\qedsymbol

\section{Equicorrelationality of generalized palindromes}\label{Alexandra}

Since \cref{Barbara} only allows for at most one self-conjugate $\laurf$-associate class within an equicorrelationality class, the following corollary, which was stated as \cref{David} in the Introduction, can now be proved.
\begin{corollary}\label{Eric}
If $f$ and $g$ are generalized palindromes that are equicorrelational, then they must be $\laurc$-associates.
\end{corollary}
\begin{proof}
If $f=0$, then it is equicorrelational to $g$ if and only if $g=0$, in which case $f$ and $g$ are clearly $\laurc$-associates.
Assume that $f\not=0$ henceforth.
By \cref{Barbara} (with $F=\C$) there is at most one self-conjugate class in the union on the right-hand side of \eqref{Rosa} of all the $\laurc$-associate classes of sequences that are equicorrelational to $f$.
Since $f$ and $g$ are both generalized palindromes, \cref{Herbert} (with $F=\C$) tells us that they must be in this one self-conjugate $\laurc$-associate class, so $f$ and $g$ are $\laurc$-associates.
\end{proof}
Now we prove the following corollary to \cref{Eric}.
\begin{corollary}\label{Julia}
Let $F$ be a subfield of $\C$ and let $f$ and $g$ be generalized palindromes with $\conj{f}=u z^j f$ and $\conj{g}=v z^k g$ for some $u, v \in \Cu$ and $j,k \in \Z$.
Suppose that $f,g \in \laurf$ and that $f$ and $g$ are equicorrelational.
Then either (i) $f=g=0$ or else (ii) $v/u = \conj{w}/w$ for some $w \in \Fu$ and $j \equiv k \pmod{2}$.
\end{corollary}
\begin{proof}
\cref{Eric} shows that $f$ and $g$ must be $\laurc$-associates, so then the conclusion follows from \cref{Laura}.
\end{proof}
Now we have the following consequence, recorded in the Introduction as \cref{Francis}.
\begin{corollary}
It is not possible for a palindrome in $\R[z,z^{-1}]$ to be equicorrelational to an antipalindrome in $\R[z,z^{-1}]$ unless both the sequences are $0$.
\end{corollary}
\begin{proof}
Suppose that a palindrome in $\R[z,z^{-1}]$ is equicorrelational to an antipalindrome in $\R[z,z^{-1}]$.  Then we apply \cref{Julia} with $u=1$ and $v=-1$ to see that there must be some $w \in \R^\times$ such that $\conj{w}/w=v/u=-1$, which is absurd.
\end{proof}

\section{Unequivocal integers}\label{Hilda}

Recall from the Introduction that we say that a binary sequence is equivocal to mean that it is nontrivially equicorrelational to some other binary sequence, and we say that a positive integer $n$ is equivocal to mean that there is an equivocal binary sequence of length $n$.  The main purpose of this section is to prove \cref{Karen}, which states that every positive multiple of an equivocal number is equivocal.  
We begin along this path with a straightforward construction that takes in two $k$-ary sequences and produces a new $k$-ary sequence whose length is the product of the lengths of the inputs.
\begin{construction}\label{Clarence}
Let $k$ be a positive integer and $\ell$ and $m$ be nonnegative integers.  Let $a$ be a $k$-ary sequence of length $\ell$ and $b$ be a $k$-ary sequence of length $m$.  Then $a(z^m) b(z)$ is a $k$-ary sequence of length $\ell m$.
\end{construction}
We now show that this construction preserves equicorrelationality in the sense that if $c$ and $d$ are sequences that are equicorrelational to $a$ and $b$, respectively, then the output sequence $c(z^m) d(z)$ is equicorrelational to $a(z^m) b(z)$.  In fact, we prove something more general in the next lemma.
\begin{lemma}\label{Victoria}
Let $m,n \in \Z$ and let $a,b,c,d \in \laurc$ such that $a$ is equicorrelational to $c$ and $b$ is equicorrelational to $d$.
Then $f(z)=a(z^m)b(z^n)$ is equicorrelational to $g(z)=c(z^m)d(z^n)$.
\end{lemma}
\begin{proof}
By the assumption of equicorrelationality, there are positive real numbers $s$
and $t$ such that $a\conj{a} = s c \conj{c}$ and $b\conj{b} = t d \conj{d}$.
For each $j \in \Z$, let $\phi_j\colon \laurc \to \laurc$ be the ring homomorphism with $\phi_j(u(z))=u(z^j)$.  Note that $\phi_j$ commutes with the conjugation map $u(z)\mapsto \conj{u(z)}$.  Thus
\begin{align*}
a(z^m)b(z^n) \conj{a(z^m)b(z^n)}
& = \phi_m(a(z) \conj{a(z)}) \phi_n(b(z) \conj{b(z)}) \\
& = \phi_m(s c(z) \conj{c(z)}) \phi_n(t d(z) \conj{d(z)}) \\
& = s t c(z^m)d(z^n) \conj{c(z^m)d(z^n)},
\end{align*}
and since $s t$ is a positive real number, we see that $a(z^m)b(z^n)$
is equicorrelational to $c(z^m)d(z^n)$.
\end{proof}
We would like to know when the equicorrelationality of $f(z)$ and $g(z)$ in \cref{Victoria} is nontrivial.  To this end, we first begin with a result that shows when $f(z)$ and $g(z)$ are associates.  Recall from the Introduction that if $f \in \laurc$, then the class of $\laurc$-associates of $f$ is denoted $[f]$.
\begin{lemma}\label{Deidre}
Let $m$ and $n$ be nonzero integers, let $a,b,c,d \in \laurc$, and let $f(z)=a(z^m)b(z^n)$ and $g(z)=c(z^m)d(z^n)$.
Then $[f]=[g]$ if and only if one of the following holds:
\begin{enumerate}[label=(\roman*)]
\item\label{Deidre-zero} $0 \in \{a,b\}$ and $0 \in \{c,d\}$; or
\item\label{Deidre-nonzero} $0 \not\in \{a,b,c,d\}$ and both $[\alpha(z^m)]=[\delta(z^n)]$ and $[\beta(z^n)]=[\gamma(z^m)]$ hold, where $\alpha,\beta,\gamma,\delta$ are the sequences such that $a=\gcd(a,c) \alpha$, $c=\gcd(a,c) \gamma$, $b=\gcd(b,d) \beta$, and $d=\gcd(b,d) \delta$.
\end{enumerate}
\end{lemma}
\begin{proof}
In case \ref{Deidre-zero}, we have $f=g=0$, so $[f]=[g]$.
If $0$ is in one and only one of $\{a,b\}$ or $\{c,d\}$, then one and only one of the two sequences $f$ and $g$ is zero, and then $[f]\not=[g]$.
So we may assume $0\not\in\{a,b,c,d\}$ for the rest of the proof.

We let $s=\gcd(a,c)$ and $t=\gcd(b,d)$; these are nonzero because $a,b,c,d$ are nonzero.
Then we define $\alpha,\beta,\gamma,\delta$ as in \ref{Deidre-nonzero}, so that $a=s \alpha$, $b=t \beta$, $c=s \gamma$, and $d=t \delta$.
We note that $[f]=[g]$ if and only if there is a $u \in \laurcu$ such that $f= u g$.
So $[f]=[g]$ if and only if there is a $u \in \laurcu$ such that $u(z) s(z^m) \alpha(z^m) t(z^n) \beta(z^n)=s(z^m) \gamma(z^m) t(z^n) \delta(z^n)$.  
Notice that $s(z^m)$ and $t(z^n)$ are nonzero because $s(z)$, $t(z)$, $m$, and $n$ are all nonzero.
Since $\laurcu$ is an integral domain, this means that $[f]=[g]$ if and only if there is a $u \in \laurcu$ such that $u(z) \alpha(z^m) \beta(z^n)=\gamma(z^m) \delta(z^n)$.
But $\alpha(z^m)$ is relatively prime to $\gamma(z^m)$ and $\beta(z^n)$ is relatively prime to $\delta(z^n)$ by \cref{Conrad}, and $\laurc$ is a unique factorization domain (since it is a Euclidean domain).
Thus, $[f]=[g]$ if and only if both $[\alpha(z^m)]=[\delta(z^n)]$ and $[\beta(z^n)]=[\gamma(z^m)]$.
\end{proof}
Now we can show when the $f$ and $g$ constructed in \cref{Victoria} are trivially equicorrelational.
\begin{lemma}\label{Hector}
Let $m$ and $n$ be nonzero integers, let $a,b,c,d \in \laurc$, and let $f(z)=a(z^m)b(z^n)$ and $g(z)=c(z^m)d(z^n)$.
Then $f$ is trivially equicorrelational to $g$ if and only if one of the following holds:
\begin{enumerate}[label=(\roman*)]
\item\label{Hector-zero} $0 \in \{a,b\}$ and $0 \in \{c,d\}$; or
\item\label{Hector-associate} $0 \not\in \{a,b,c,d\}$ and both $[\alpha(z^m)]=[\delta(z^n)]$ and $[\beta(z^n)]=[\gamma(z^m)]$ hold, where $\alpha,\beta,\gamma,\delta$ are the sequences such that $a=\gcd(a,c) \alpha$, $c=\gcd(a,c) \gamma$, $b=\gcd(b,d) \beta$, and $d=\gcd(b,d) \delta$.
\item\label{Hector-conjugate} $0 \not\in \{a,b,c,d\}$ and both $[A(z^m)]=[\Delta(z^n)]$ and $[B(z^n)]=[\Gamma(z^m)]$ hold, where $A,B,\Gamma,\Delta$ are the sequences such that $a=\gcd(a,\conj{c}) A$, $\conj{c}=\gcd(a,\conj{c}) \Gamma$, $b=\gcd(b,\conj{d}) B$, and $\conj{d}=\gcd(b,\conj{d}) \Delta$.
\end{enumerate}
\end{lemma}
\begin{proof}
By the definition of trivial equicorrelationality, $f$ and $g$ are trivially equicorrelational if and only if either $[f]=[g]$ or $[f]=[\conj{g}]$.  \cref{Deidre} says that $[f]=[g]$ if and only if either \ref{Hector-zero} or \ref{Hector-associate} holds.  Since $\conj{g(z)}=\conj{c}(z^m)\conj{d}(z^n)$, \cref{Deidre} shows that $[f]=[\conj{g}]$ if and only if either \ref{Hector-zero} or \ref{Hector-conjugate} holds.
\end{proof}
\cref{Hector} applies to sequences $a$, $b$, $c$, and $d$ that need not be binary (nor, more generally, $k$-ary for some $k$), and even if these four sequences are binary (or $k$-ary), the parameters $m$ and $n$ might be such that the combined sequences $f$ and $g$ are not binary (or $k$-ary).
In many practical scenarios, we would want to constrain $a$, $b$, $c$, $d$, $m$, and $n$ so as to produce binary (or $k$-ary) $f$ and $g$, and we examine such situations in the following result, which makes use of \cref{Clarence}.
\begin{proposition}\label{Milton}
Let $k$, $\ell$, and $m$ be a positive integers.
Let $a$ and $c$ be equicorrelational $k$-ary sequences of length $\ell$.
Let $b$ and $d$ be equicorrelational $k$-ary sequences of length $m$.
Then $f(z)=a(z^m)b(z)$ and $g(z)=c(z^m)d(z)$ are equicorrelational $k$-ary sequences of length $\ell m$.
Furthermore, $f$ is trivially equicorrelational to $g$ if and only if at least one of the following two conditions holds:
\begin{enumerate}[label=(\roman*)]
\item\label{Milton-associate} $[a]=[c]$ and $[b]=[d]$; or
\item\label{Milton-conjugate} $[a]=[\conj{c}]$ and $[b]=[\conj{d}]$.
\end{enumerate}
In particular, if $a$ is nontrivially equicorrelational to $c$ or if $b$ is nontrivially equicorrelational to $d$, then $f$ is certainly nontrivially equicorrelational to $g$.
\end{proposition}
\begin{proof}
The fact that $f$ and $g$ are $k$-ary sequences of length $\ell m$ comes from \cref{Clarence}, and the fact that they are equicorrelational comes from \cref{Victoria}.
Notice that $a$, $b$, $c$, and $d$ are all nonzero because of the given lengths of these sequences.

If case \ref{Milton-associate} of this proposition holds, then case \ref{Hector-associate} of \cref{Hector} holds because the $\alpha$, $\beta$, $\gamma$, and $\delta$ defined there are all units, so then \cref{Dorothy} makes $\alpha(z^m)$ and $\gamma(z^m)$ units, so that $[\alpha(z^m)]=[\delta(z)]$ and $[\beta(z)]=[\gamma(z^m)]$.
If case \ref{Milton-conjugate} of this proposition holds, then case \ref{Hector-conjugate} of \cref{Hector} holds, because the $A$, $B$, $\Gamma$, and $\Delta$ defined there are all units, so then \cref{Dorothy} makes $A(z^m)$ and $\Gamma(z^m)$ units, so that $[A(z^m)]=[\Delta(z)]$ and $[B(z)]=[\Gamma(z^m)]$.
So \cref{Hector} shows that $f$ and $g$ are trivially equicorrelational if either condition \ref{Milton-associate} or \ref{Milton-conjugate} of this proposition holds.

Conversely, suppose that $f$ and $g$ are trivially equicorrelational.
Then \cref{Hector} applies, and since $a$, $b$, $c$, and $d$ are nonzero, we must be in either case \ref{Hector-associate} or case \ref{Hector-conjugate} of \cref{Hector}.
If \cref{Hector}\ref{Hector-associate} holds, then define $\alpha$, $\beta$, $\gamma$, and $\delta$ as they are there, and we have $[\alpha(z^m)]=[\delta(z)]$ and $[\beta(z)]=[\gamma(z^m)]$.
Otherwise, \cref{Hector}\ref{Hector-conjugate} holds, and then define $A$, $B$, $\Gamma$, and $\Delta$ as they are there, and we have $[A(z^m)]=[\Delta(z)]$ and $[B(z)]=[\Gamma(z^m)]$.
In the former case, set $\fa=\alpha$, $\fb=\beta$, $\fc=\gamma$, and $\fd=\delta$, and in the latter, set $\fa=A$, $\fb=B$, $\fc=\Gamma$, and $\fd=\Delta$.
So in either case we have $[\fa(z^m)]=[\fd(z)]$ and $[\fb(z)]=[\fc(z^m)]$.
Since associates have the same length, we may use \cref{Samuel} to conclude that $(-1+\len\fa)m+1=\len\fd$ and $(-1+\len\fc)m+1=\len\fb$.
By \cref{Lester} and our given assumptions, we have $\len b=\len \conj{b} =\len d =\len \conj{d}=m > 0$.
Since $\fb$ is a divisor of either $b$ or $\conj{b}$ and since $\fd$ is a divisor of either $d$ or $\conj{d}$ (and $m> 0$ makes all these sequences nonzero), we know that $0 < \len \fb \leq m$ and $0 < \len \fd \leq m$.
So $\len\fb$ and $\len\fd$ are positive integers that are both $1$ modulo $m$ and not greater than $m$.
Hence $\len\fb=\len\fd=1$, making $\fb$ and $\fd$ units.
But $[\fa(z^m)]=[\fd(z)]$ and $[\fb(z)]=[\fc(z^m)]$, and associates have the same length, so we may use \cref{Samuel} to conclude that $\len\fa=\len\fc=1$, and so $\fa$ and $\fc$ are also units.
If we are in case \ref{Hector-associate} of \cref{Hector}, this makes $\alpha$, $\beta$, $\gamma$, and $\delta$ there units, and this implies that $[a]=[c]$ and $[b]=[d]$, so we are in case \ref{Milton-associate} of this proposition.
If we are in case \ref{Hector-conjugate} of \cref{Hector}, this makes $A$, $B$, $\Gamma$, and $\Delta$ there units, and this implies that $[a]=[\conj{c}]$ and $[b]=[\conj{d}]$, so we are in case \ref{Milton-conjugate} of this proposition.
This completes the proof of the claim that $[f]=[g]$ if and only if we are in either case \ref{Milton-associate} or \ref{Milton-conjugate} of this proposition, from which the final claim of the proposition follows.
\end{proof}

Now we are ready to restate and prove \cref{Karen}.
\begin{proposition}\label{Ethel}
Let $m$, $n$ be positive integers such that $m | n$. If $m$ is equivocal, then $n$ is equivocal.
\end{proposition}
\begin{proof}
Suppose that $m$ is equivocal, and so there are nontrivially equicorrelational binary sequences $b$ and $d$ of length $m$.
Let $a=c$ be some binary sequence of length $n/m$.
Then $f(z)=a(z^m) b(z)$ and $g(z)=c(z^m) d(z)$ are nontrivially equicorrelational binary sequences of length $n$ by \cref{Milton}.
\end{proof}

\begin{remark}
In the proof of \cref{Ethel}, if one uses $a(z)=c(z)=1+z+\cdots+z^{n/m-1}$ and if we use $u|v$ to denote the concatenation of two sequences $u$ and $v$, then the proof could be summarized by saying that if $b$ and $d$ are nontrivially equicorrelational, then
\[
\underbrace{b|b|\cdots|b}_{\text{$n/m$ copies}} \text{ and } \underbrace{d|d|\cdots|d}_{\text{$n/m$ copies}}
\]
are nontrivially equicorrelational.
Although we could have proved \cref{Ethel} more quickly by confining ourselves to this basic construction, we proved the more general results presented in \cref{Hector} and \cref{Milton} in order to show that there are many ways in which nontrivial equicorrelationality of longer sequences can arise from nontrivial equicorrelationality of shorter sequences.  If $n$ is an equivocal number, it would be interesting to see how many of the nontrivially equicorrelational pairs of binary sequences of length $n$ can be accounted for via \cref{Milton} as arising from nontrivial equicorrelationality of sequences of some smaller length $m$ with $m\mid n$.
\end{remark}

\section*{Acknowledgments}

The authors thank Showmic Islam, Christina Koch, Andrew Owen, and Mats Rynge of the Open Science Grid for their assistance.  The authors thank anonymous reviewers for encouragement and comments which helped them to improve the paper.

\newcommand{\etalchar}[1]{$^{#1}$}

\end{document}